\documentclass[conference]{IEEEtran}
\ifCLASSINFOpdf
\else
\fi
\usepackage{url}

\usepackage{amsmath}
\usepackage{amsthm}
\usepackage{amssymb}

\theoremstyle{plain}
\newtheorem{theorem}{Theorem}

\newtheorem{proposition}[theorem]{Proposition}

\newtheorem{definition}[theorem]{Definition}

\newtheorem{replacements}{Replacements}

\newcommand{\abs}[1]{\left\lvert#1\right\rvert}

\DeclareMathOperator{\Dom}{dom}
\newcommand{\rest}[2]{#1\!\!\restriction_{#2}}
\newcommand{\reste}[2]{#1\restriction_{#2}}

\newcommand{\N}{\mathbb{N}}%
\newcommand{\Q}{\mathbb{Q}}%
\newcommand{\R}{\mathbb{R}}%
\newcommand{\X}{\{0,1\}^*}%
\newcommand{\XI}{\{0,1\}^\infty}%
\begin{document}
\title{Phase Transition and Strong Predictability}

\author{\IEEEauthorblockN{Kohtaro Tadaki}
\IEEEauthorblockA{Research and Development Initiative, Chuo University\\
1-13-27 Kasuga, Bunkyo-ku, Tokyo 112-8551, Japan\\
Email: tadaki@kc.chuo-u.ac.jp\quad
WWW: http://www2.odn.ne.jp/tadaki/}}

\maketitle

\begin{abstract}
\boldmath
The statistical mechanical interpretation of algorithmic information theory (AIT, for short) was introduced and developed in our former work
[K.~Tadaki, Local Proceedings of CiE 2008, pp.425--434, 2008],
where we introduced
the notion of thermodynamic quantities into AIT.
These quantities are real functions of
temperature $T>0$.
The values of all the thermodynamic quantities diverge when
$T$ exceeds $1$.
This phenomenon corresponds to phase transition in statistical mechanics.
In this paper we introduce the notion of strong predictability for an infinite binary sequence and then apply it to the partition function $Z(T)$,
which is one of the thermodynamic quantities in AIT.
We then reveal a new computational aspect of the phase transition in AIT by showing the critical difference of the behavior of $Z(T)$ between $T=1$ and $T<1$
in terms of the strong predictability for the base-two expansion of $Z(T)$.
\end{abstract}
\IEEEpeerreviewmaketitle

\section{Introduction}
\label{introduction}

Algorithmic information theory (AIT, for short) is a framework for applying
information-theoretic and probabilistic ideas to
computability theory.
One of the primary concepts of AIT is the \emph{program-size complexity} (or \emph{Kolmogorov complexity}) $H(x)$ of a finite binary string $x$,
which is defined as the length of the shortest binary
program
for a universal decoding algorithm $U$, called an \emph{optimal prefix-free machine}, to output $x$.
By the definition, $H(x)$ is thought to represent the amount of randomness contained in a finite binary string $x$.
In particular, the notion of program-size complexity plays a crucial role in characterizing the \emph{randomness} of an infinite binary sequence, or equivalently, a real.
In \cite{C75} Chaitin introduced
the $\Omega$ number
as a concrete example of random real.
The first $n$ bits of the base-two expansion of $\Omega$ solve the halting problem of $U$ for inputs of length at most $n$.
By this property,
$\Omega$ is shown to be a random real,
and
plays a central role in the development of AIT.

In this paper, we study the \emph{statistical mechanical interpretation} of AIT.
In a series of works \cite{T08CiE,T09LFCS,T09ITW,T10JPCS,T11ITW,T12WTCS}, we introduced and developed
this particular subject of AIT.
First, in \cite{T08CiE} we introduced the \emph{thermodynamic quantities} at temperature $T$,
such as partition function $Z(T)$, free energy $F(T)$, energy $E(T)$, statistical mechanical entropy $S(T)$, and specific heat $C(T)$, into AIT.
These quantities are real functions of a real argument $T>0$, and are introduced
in the following manner:
Let $X$ be a complete set of energy eigenstates of
a quantum system
and $E_x$ the energy of an energy eigenstate $x$ of the quantum system.
In \cite{T08CiE} we introduced
thermodynamic quantities
into AIT by performing Replacements~\ref{CS06} below for the corresponding thermodynamic quantities
in statistical mechanics.

\begin{replacements}\label{CS06}\hfill
\vspace*{-1mm}
\begin{enumerate}
  \item Replace the complete set $X$ of energy eigenstates $x$
    by the set $\Dom U$ of all programs $p$ for $U$.
  \item Replace the energy $E_x$ of an energy eigenstate $x$
    by the length $\abs{p}$ of a program $p$.
  \item Set the Boltzmann Constant $k_{\mathrm{B}}$ to $1/\ln 2$.\hfill\IEEEQED
\end{enumerate}
\end{replacements}

\vspace*{-1mm}

For example, in statistical mechanics, the partition function $Z_{\mathrm{sm}}(T)$ at temperature $T$ is given by
\vspace*{-0.5mm}
\begin{equation*}%
  Z_{\mathrm{sm}}(T)=\sum_{x\in X}e^{-\frac{E_x}{k_{\mathrm{B}}T}}.
\end{equation*}
\vspace*{-2.0mm}\\
Thus, based on Replacements~\ref{CS06}, the partition function $Z(T)$
in AIT is defined as
\vspace*{-1mm}
\begin{equation}\label{def_Z(T)}
  Z(T)=\sum_{p\in\Dom U}2^{-\frac{\abs{p}}{T}}.
\end{equation}
\vspace*{-2.0mm}\\
In general, the thermodynamic quantities in AIT are variants of Chaitin $\Omega$ number.
In fact,
in the case of $T=1$, $Z(1)$ is precisely Chaitin $\Omega$ number.%
\footnote{To be precise, the partition function is not a thermodynamic quantity but a statistical mechanical quantity.}

In \cite{T08CiE} we then proved that if the temperature $T$ is a computable real with $0<T<1$ then,
for each of the thermodynamic quantities $Z(T)$, $F(T)$, $E(T)$, $S(T)$, and $C(T)$ in AIT,
the partial randomness of its value equals to $T$,
where the notion of \emph{partial randomness} is a stronger representation of the compression rate by means of program-size complexity.
Thus,
the temperature $T$ plays a role as the partial randomness (and therefore the compression rate) of all the thermodynamic quantities
in the statistical mechanical interpretation of AIT.
In \cite{T08CiE}
we further showed that
the temperature $T$ plays a role as the partial randomness of the temperature $T$ itself,
which is a thermodynamic quantity of itself in thermodynamics or statistical mechanics.
Namely, we proved the \emph{fixed point theorem on partial randomness},%
\footnote{The fixed point theorem on partial randomness is called a fixed point theorem on compression rate in \cite{T08CiE}.}
which states that, for every $T\in(0,1)$, if the value of partition function $Z(T)$ at temperature $T$ is a computable real,
then the partial randomness of $T$ equals to $T$, and therefore the compression rate of $T$ equals to $T$, i.e.,
$\lim_{n\to\infty}H(\rest{T}{n})/n=T$,
where $\rest{T}{n}$ is the first $n$ bits of the base-two expansion of the real $T$.

In our second work \cite{T09LFCS} on the statistical mechanical interpretation of AIT,
we showed that a fixed point theorem of the same form
as for $Z(T)$ holds also for each of $F(T)$, $E(T)$, and $S(T)$.
In the third work \cite{T09ITW}, we further unlocked the properties of
the fixed points on partial randomness
by introducing the notion of composition of prefix-free machines into AIT,
which corresponds to the notion of composition of systems in normal statistical mechanics.
In the work \cite{T10JPCS} we developed a total statistical mechanical interpretation of AIT which attains a perfect correspondence to normal statistical mechanics,
by making an argument on the same level of mathematical strictness as normal statistical mechanics in physics.
We did this by identifying a \emph{microcanonical ensemble} in AIT.
This identification clarifies the meaning of the thermodynamic quantities of AIT.

Our first work
\cite{T08CiE} showed that the values of all the thermodynamic quantities in AIT diverge when the temperature $T$ exceeds $1$.
This phenomenon might be regarded as some sort of \emph{phase transition} in statistical mechanics.
In the work \cite{T12WTCS} we
revealed a computational aspect of the phase transition in AIT.
The notion of \emph{weak truth-table reducibility} plays an important role in recursion theory \cite{N09,DH10}.
In the work \cite{T12WTCS} we introduced an elaboration of this notion, called \emph{reducibility in query size $f$}.
This elaboration enables us to deal with the notion of asymptotic behavior of computation in a manner like in computational complexity theory,
while staying in computability theory.
We applied the elaboration to the relation between $Z(T)$ and $\Dom U$,
where the latter is the set of all halting inputs for the optimal prefix-free machine $U$,
i.e.,
the \emph{halting problem}.
We then revealed the critical difference of the behavior of $Z(T)$ between $T=1$ and $T<1$ in relation to $\Dom U$.
Namely, we revealed the phase transition between the \emph{unidirectionality} at $T=1$ and the \emph{bidirectionality} at $T<1$
in the reduction between $Z(T)$ and $\Dom U$.
This critical phenomenon cannot be captured by the original notion of weak truth-table reducibility.

In this paper, we reveal another computational aspect of the phase transition in AIT between $T=1$ and $T<1$.
We introduce the notion of \emph{strong predictability} for an infinite binary sequence.
Let $X=b_1b_2b_3\dotsc$ be an infinite binary sequence with each $b_i\in\{0,1\}$.
The strong predictability of $X$ is the existence of the computational procedure which, given any prefix $b_1\dots b_n$ of $X$, can predict the next bit $b_{n+1}$ in $X$
with unfailing accuracy,
where the suspension of an individual prediction for the next bit is allowed to make sure that the whole predictions are error-free.
We introduce three types of strong predictability,
\emph{finite-state strong predictability}, \emph{total strong predictability}, and \emph{strong predictability},
which differ with respect to computational ability.
We apply them to the base-two expansion of $Z(T)$.
On the one hand,
we show that the base-two expansion of $Z(T)$ is not strongly predictable at $T=1$ in the sense of any of these three types of strong predictability.
On the other hand,
we show that it is strongly predictable in the sense of all of the three types in the case where $T$ is computable real with $T<1$.
In this manner, we reveal a new aspect of the phase transition in AIT between $T=1$ and $T<1$.
\vspace*{-1.0mm}

\section{Preliminaries}
\label{preliminaries}

We start with some notation and definitions which will be used in this paper.
For any set $S$ we denote by $\#S$ the cardinality of $S$.
$\N=\left\{0,1,2,3,\dotsc\right\}$ is the set of natural numbers, and $\N^+$ is the set of positive integers.
$\Q$ is the set of rationals, and $\R$ is the set of reals.
$\X =
\left\{
   \lambda,0,1,00,01,10,11,000,\dotsc
\right\}$
is the set of finite binary strings, where $\lambda$ denotes the \textit{empty string}, and $\X$ is ordered as indicated.
We identify any string in $\X$ with a natural number in this order.
For any $x\in \X$, $\abs{x}$ is the \textit{length} of $x$.
A subset $S$ of $\X$ is called \emph{prefix-free} if no string in $S$ is a prefix of another string in $S$.

We denote by $\XI$ the set of infinite binary sequences, where an infinite binary sequence is infinite to the right but finite to the left.
Let $X\in\XI$.
For any $n\in\N^+$, we denote the $n$th bit of $X$ by $X(n)$.
For any $n\in\N$, we denote the first $n$ bits of  $X$ by $\rest{X}{n}\in\X$.
Namely, $\rest{X}{0}=\lambda$, and $\rest{X}{n}=X(1)X(2)\dots X(n)$ for every $n\in\N^+$.

For any real $\alpha$, we denote by $\lfloor \alpha \rfloor$ the greatest integer less than or equal to $\alpha$.
When we mention a real $\alpha$ as an infinite binary sequence,
we are considering the base-two expansion of the fractional part $\alpha - \lfloor \alpha \rfloor$ of the real $\alpha$ with infinitely many zeros.
Thus, for any real $\alpha$, $\rest{\alpha}{n}$ and $\alpha(n)$ denote $\rest{X}{n}$ and $X(n)$, respectively,
where $X$ is the unique infinite binary sequence such that $\alpha - \lfloor \alpha \rfloor=0.X$ and $X$ contains infinitely many zeros.

A function $f\colon\N\to\X$ or $f\colon\N\to\Q$ is called \emph{computable}
if there exists a deterministic Turing machine which on every input $n\in\N$ halts and outputs $f(n)$.
A real $\alpha$ is called \emph{computable} if there exists a computable function $f\colon\N\to\Q$ such that $\abs{\alpha-f(n)} < 2^{-n}$ for all $n\in\N$.
We say that $X\in\XI$ is \emph{computable} if the mapping $\N\ni n\mapsto\rest{X}{n}$ is a computable function,
which is equivalent to that the real $0.X$ in base-two notation is computable.

Let $S$ and $T$ be any sets. We say that $f\colon S\to T$ is a \emph{partial function} if $f$ is a function whose domain is a subset of $S$ and whose range is $T$.
The domain of a partial function $f\colon S\to T$ is denoted by $\Dom f$.
A \emph{partial computable function} $f\colon\X\to\X$ is a partial function $f\colon\X\to\X$ for which
there exists a deterministic Turing machine $M$ such that (i) on every input $x\in\X$, $M$ halts if and only of $x\in\Dom f$,
and (ii) on every input $x\in\Dom f$, $M$ outputs $f(x)$.
We write ``c.e.'' instead of ``computably enumerable.''

\vspace*{-1.5mm}

\subsection{Algorithmic Information Theory}
\label{ait}

\vspace*{-0.5mm}

In the following we concisely review some definitions and results of
AIT
\cite{C75,C87b,N09,DH10}.
A \emph{prefix-free machine} is a partial computable function $M\colon \X\to \X$ such that $\Dom M$ is
prefix-free.
For each prefix-free machine $M$ and each $x\in \X$, $H_M(x)$ is defined by
$H_M(x) =
\min
\left\{\,
  \abs{p}\,\big|\;p \in \X\>\&\>M(p)=x
\,\right\}$
(may be $\infty$).
A prefix-free machine $U$ is called \emph{optimal} if for each prefix-free machine $M$ there exists $d\in\N$ with the following property;
if $p\in\Dom M$, then there is $q\in\Dom U$ for which $U(q)=M(p)$ and $\abs{q}\le\abs{p}+d$.
It is then easy to see that there exists an optimal prefix-free machine.
We choose a particular optimal prefix-free machine $U$ as the standard one for use, and define $H(x)$ as $H_U(x)$, which is referred to as
the \emph{program-size complexity} of $x$ or the \emph{Kolmogorov complexity} of $x$.

Chaitin~\cite{C75} introduced $\Omega$ number by
$\Omega=\sum_{p\in\Dom U}2^{-\abs{p}}$.
Since $\Dom U$ is prefix-free, $\Omega$ converges and $0<\Omega\le 1$.
For any $X\in\XI$, we say that $X$ is \emph{weakly Chaitin random} if there exists $c\in\N$ such that $n-c\le H(\rest{X}{n})$ for all $n\in\N^+$ \cite{C75,C87b}.
Chaitin \cite{C75} showed that $\Omega$ is weakly Chaitin random.
Therefore $0<\Omega<1$.

\vspace*{-1.5mm}

\subsection{Partial Randomness}
\label{PR}

\vspace*{-0.5mm}

In the work \cite{T02}, we generalized the notion of the randomness of a real
so that the \emph{partial randomness} of a real can be characterized by a real $T$ with $0\le T\le 1$ as follows.

\vspace*{-1mm}

\begin{definition}[Tadaki~\cite{T02}]
Let $T\in[0,1]$ and let $X\in\XI$.
We say that $X$ is \emph{weakly Chaitin $T$-random} if
there exists $c\in\N$ such that, for all $n\in\N^+$, $Tn-c \le H(\rest{X}{n})$.
\hfill\IEEEQED
\end{definition}

\vspace*{-1mm}

In the case of $T=1$, the weak Chaitin $T$-randomness results in the weak Chaitin randomness.

\begin{definition}[Tadaki~\cite{T12WTCS}]
Let $T\in[0,1]$ and let $X\in\XI$.
We say that $X$ is \emph{strictly $T$-compressible} if there exists $d\in\N$ such that, for all $n\in\N^+$, $H(\rest{X}{n})\le Tn+d$.
We say that $X$ is \emph{strictly Chaitin $T$-random} if $X$ is both weakly Chaitin $T$-random and strictly $T$-compressible. 
\hfill\IEEEQED
\end{definition}

\vspace*{-1mm}

In the work \cite{T02}, we generalized Chaitin $\Omega$ number to $Z(T)$ as follows.
For each real $T>0$,
the \textit{partition function} $Z(T)$ at temperature $T$ is defined by
the equation \eqref{def_Z(T)}.
Thus, $Z(1)=\Omega$.
If $0<T\le 1$, then $Z(T)$ converges and $0<Z(T)<1$, since $Z(T)\le \Omega<1$.
The following theorem holds for $Z(T)$.

\begin{theorem}[Tadaki \cite{T02,T12WTCS}]\label{ZVTwCTr-strictTcb}
Let $T\in\R$.
\begin{enumerate}
\vspace*{-1mm}
  \item If $0<T<1$ and $T$ is computable, then $Z(T)$ is strictly Chaitin $T$-random.
  \item If $1<T$, then $Z(T)$ diverges to $\infty$.\hfill\IEEEQED
\end{enumerate}
\end{theorem}

\vspace*{-1mm}

This theorem shows some aspect of the phase transition of the behavior of $Z(T)$
when the temperature $T$ exceeds $1$.

\vspace*{-1.0mm}

\subsection{Martingales}
\label{MG}

\vspace*{-0.5mm}

In this subsection we review the notion of \emph{martingale}.
Compared with the notion of strong predictability which is introduced in this paper, the predictability based on martingale is weak one.
We refer the reader to
Nies \cite[Chapter 7]{N09}
for the notions and results of this subsection.

A martingale $B$ is a betting strategy.
Imagine a gambler in a casino is presented with prefixes of an infinite binary sequence $X$ in ascending order.
So far she has been seen a prefix $x$ of $X$, and her current capital is $B(x)\ge 0$.
She bets an amount $\alpha$ with $0\le \alpha\le B(x)$ on her prediction that the next bit will be $0$, say.
Then the bit is revealed. If she was right, she wins $\alpha$, else she loses $\alpha$.
Thus, $B(x0)=B(x)+\alpha$ and $B(x1)=B(x)-\alpha$, and hence $B(x0)+B(x1)=2 B(x)$.
The same considerations apply if she bets that the next bit will be $1$.
These considerations result in the following definition.

\begin{definition}[Martingale]
A \emph{martingale} is a function $B\colon\X\to[0,\infty)$
such that $B(x0)+B(x1)=2 B(x)$ for every $x\in\X$.
For any $X\in\XI$, we say that the martingale $B$ \emph{succeeds} on $X$ if
the capital it reaches along $X$ is unbounded, i.e., $\sup\{B(\rest{X}{n})\mid n\in\N\}=\infty$.
\hfill\IEEEQED
\end{definition}

\vspace*{-1mm}

For any subset $S$ of $\X\times\Q$, we say that $S$ is \emph{computably enumerable} (\emph{c.e.}, for short) if
there exists a deterministic Turing machine $M$ such that, on every input $s\in\X\times\Q$, $M$ halts if and only if $s\in S$.

\begin{definition}[C.E.~Martingale]
A martingale $B$ is called \emph{computably enumerable} if the set $\{(x,q)\in\X\times\Q\mid q<B(x)\}$ is c.e.
\hfill\IEEEQED
\end{definition}

\vspace*{-1mm}

\begin{theorem}
For every $X\in\XI$, no c.e.~martingale succeeds on $X$ if and only if $X$ is weakly Chaitin random.
\hfill\IEEEQED
\end{theorem}

\vspace*{-2mm}

For any subset $S$ of $\X\times\Q$, we say that $S$ is \emph{computable} if
there exists a deterministic Turing machine $M$ such that, on every input $s\in\X\times\Q$,
(i) $M$ halts and (ii) $M$ outputs $1$ if $s\in S$ and $0$ otherwise.

\begin{definition}[Computable Randomness]
A martingale $B$ is called \emph{computable} if
the set $\{(x,q)\in\X\times\Q\mid q<B(x)\}$ is computable.
For any $X\in\XI$, we say that $X$ is \emph{computably random} if no computable martingale succeeds on $X$.
\hfill\IEEEQED
\end{definition}

\begin{definition}[Partial Computable Martingale]
A \emph{partial computable martingale} is a partial computable function $B\colon\X\to\Q\cap[0,\infty)$ such that
$\Dom B$ is closed under prefixes, and for each $x\in\Dom B$, $B(x0)$ is defined iff $B(x1)$ is defined,
in which case
$B(x0)+B(x1)=2 B(x)$ holds.
\hfill\IEEEQED
\end{definition}

\begin{definition}[Partial Computable Randomness]
Let $B$ be a partial computable martingale and $X\in\XI$.
We say that $B$ \emph{succeeds} on $X$ if $B(\rest{X}{n})$ is defined for all $n\in\N$ and $\sup\{B(\rest{X}{n})\mid n\in\N\}=\infty$.
We say that $X$ is \emph{partial computably random} if no partial computable martingale succeeds on $X$.
\hfill\IEEEQED
\end{definition}

\begin{theorem}\label{wCpcr_pcrcr}
Let $X\in\XI$.
\vspace*{-2mm}
\begin{enumerate}
  \item If $X$ is weakly Chaitin random then $X$ is partial computably random.
  \item If $X$ is partial computably random then $X$ is computably random.\hfill\IEEEQED
\end{enumerate}
\end{theorem}

\vspace*{-2mm}

The converse direction of each of the implications (i) and (ii) of Theorem~\ref{wCpcr_pcrcr} fails.
\vspace*{-1mm}

\section{Non Strong Predictability at $T=1$}
\label{T=1}

\vspace*{-1.0mm}

The main result in this section is Theorem~\ref{main3}, which shows that partial computable randomness implies non strong predictability.
For intelligibility
we first show an easier result, Theorem~\ref{main2}, which says that computable randomness implies non total strong predictability.

\begin{definition}[Total Strong Predictability]\label{def_tsp}
For any $X\in\XI$, we say that $X$ is
\emph{total strongly predictable} if there exists a computable function $F\colon\X\to\{0,1,N\}$
for which the following two conditions hold:
\begin{enumerate}
\vspace*{-1mm}
\item For every $n\in\N$, if $F(\rest{X}{n})\neq N$ then
  $F(\rest{X}{n})=X(n+1)$.
\item The set $\{n\in\N\mid F(\rest{X}{n})\neq N\}$ is infinite.\hfill\IEEEQED
\end{enumerate}
\end{definition}

\vspace*{-2mm}

In the above definition, the letter $N$ outputted by $F$ on the input $\rest{X}{n}$ means that the prediction of the next bit $X(n+1)$ is suspended. 

\begin{theorem}\label{main2}
For every $X\in\XI$, if $X$ is computably random then $X$ is not total strongly predictable.
\end{theorem}

\vspace*{-2mm}

\begin{proof}
We show the contraposition of Theorem~\ref{main2}.
For that purpose, suppose that $X$ is total strongly predictable.
Then there exists a computable function $F\colon\X\to\{0,1,N\}$ which satisfies the conditions (i) and (ii) of Definition~\ref{def_tsp}.
We define a function $B\colon\X\to\N$ recursively as follows:
First $B(\lambda)$ is defined as $1$.
Then, for any $x\in\X$, $B(x 0)$ is defined by
\vspace*{-3mm}
\begin{equation*}
B(x 0)=
\left\{
\begin{array}{ll}
B(x) & \text{if }F(x)=N, \\
2B(x) & \text{if }F(x)=0, \\
0 & \text{otherwise},
\end{array}
\right.
\end{equation*}
\vspace*{-1.5mm}\\
and then $B(x 1)$ is defined by
$B(x 1)=2B(x)-B(x 0)$.
It follows that $B\colon\X\to\N$ is a computable function and 
\begin{equation*}
  B(x 0)+B(x 1)=2B(x)
\end{equation*}
for every $x\in\X$.
Thus $B$ is a computable martingale.
On the other hand, it is easy to see that
\vspace*{-1mm}
\begin{equation*}
  B(\rest{X}{n})=2^{\#\{m\in\N\,\mid\,m<n\text{ \& }F(\reste{X}{m})\neq N\}}
\end{equation*}
\vspace*{-3.5mm}\\
for every $n\in\N$.
Since the set $\{n\in\N\mid F(\rest{X}{n})\neq N\}$ is infinite, it follows that
$\lim_{n\to\infty} B(\rest{X}{n})=\infty$.
Therefore, $X$ is not computably random, as desired.
\end{proof}

\begin{definition}[Strong Predictability]\label{def_sp}
For any $X\in\XI$, we say that $X$ is
\emph{strongly predictable} if there exists a partial computable function $F\colon\X\to\{0,1,N\}$ for which
the following three conditions hold:
\begin{enumerate}
\vspace*{-1mm}
\item For every $n\in\N$, $F(\rest{X}{n})$ is defined.
\item For every $n\in\N$, if $F(\rest{X}{n})\neq N$ then
  $F(\rest{X}{n})=X(n+1)$.
\item The set $\{n\in\N\mid F(\rest{X}{n})\neq N\}$ is infinite.\hfill\IEEEQED
\end{enumerate}
\end{definition}

\vspace*{-1mm}

Obviously, the following holds.

\begin{proposition}\label{fact1}
For every $X\in\XI$, if $X$ is total strongly predictable then $X$ is strongly predictable.
\hfill\IEEEQED
\end{proposition}

\begin{theorem}\label{main3}
For every $X\in\XI$, if $X$ is partial computably random then $X$ is not strongly predictable.
\end{theorem}

\vspace*{-2mm}

\begin{proof}
We show the contraposition of Theorem~\ref{main3}.
For that purpose, suppose that $X$ is strongly predictable.
Then there exists a partial computable function $F\colon\X\to\{0,1,N\}$ which satisfies the conditions (i), (ii), and (iii) of Definition~\ref{def_sp}.
We define a partial function $B\colon\X\to\N$ recursively as follows:
First $B(\lambda)$ is defined as $1$.
Then, for any $x\in\X$, $B(x 0)$ is defined by
\begin{equation*}
B(x 0)=
\left\{ \begin{array}{ll}
B(x) & \text{if }F(x)=N, \\
2B(x) & \text{if }F(x)=0, \\
0 & \text{if }F(x)=1, \\
\text{undefined} & \text{if $F(x)$ is undefined},
\end{array} \right.
\end{equation*}
and then $B(x 1)$ is defined by
\begin{equation*}
B(x 1)=
\left\{ \begin{array}{ll}
2B(x)-B(x 0) & \text{if $B(x 0)$ is defined}, \\
\text{undefined} & \text{otherwise}.
\end{array} \right.
\end{equation*}
It follows that $B\colon\X\to\N$ is a partial computable function such that
\vspace*{-1mm}
\begin{enumerate}
\item $\Dom B$ is closed under prefixes,
\item for every $x\in\Dom B$, $x 0\in\Dom B$ if and only if $x 1\in\Dom B$, and
\item for every $x\in\X$, if $x,x 0, x 1\in\Dom B$ then $B(x 0)+B(x 1)=2B(x)$.
\end{enumerate}
\vspace*{-1mm}
Thus $B$ is a partial computable martingale.
On the other hand, it is easy to see that, for every $n\in\N$, $B(\rest{X}{n})$ is defined and
$B(\rest{X}{n})=2^{\#\{m\in\N\,\mid\,m<n\text{ \& }F(\reste{X}{m})\neq N\}}$
Since the set $\{n\in\N\mid F(\rest{X}{n})\neq N\}$ is infinite, it follows that
$\lim_{n\to\infty} B(\rest{X}{n})=\infty$.
Therefore, $X$ is not partial computably random, as desired.
\end{proof}

\begin{theorem}\label{wCrnonsp}
For every $X\in\XI$, if $X$ is weakly Chaitin random then $X$ is not strongly predictable.
\end{theorem}

\vspace*{-2mm}

\begin{proof}
The result follows immediately from (i) of Theorem~\ref{wCpcr_pcrcr} and Theorem~\ref{main3}.
\end{proof}

Thus, since $Z(1)$, i.e., $\Omega$, is weakly Chaitin random, we have the following.

\begin{theorem}\label{ZT=1nonsp}
$Z(1)$ is not strongly predictable.
\hfill\IEEEQED
\end{theorem}

\section{Strong Predictability on $T<1$}
\label{T<1}

In this section, we introduce the notion of \emph{finite-state strong predictability}.
For that purpose, we first introduce the notion of \emph{finite automaton with outputs}.
This is just a deterministic finite automaton
whose output is determined, depending only on its final state.
The formal definitions are as follows.

\begin{definition}[Finite Automaton with Outputs]\label{def_FAO}
A finite automaton with outputs is a $6$-tuple $(Q,\Sigma,\delta,q_0,\Gamma,f)$, where
\vspace*{-1mm}
\begin{enumerate}
\item $Q$ is a finite set called the \emph{states},
\item $\Sigma$ is a finite set called the \emph{input alphabet},
\item $\delta\colon Q\times\Sigma\to Q$ is the \emph{transition function},
\item $q_0\in Q$ is the \emph{initial state},
\item $\Gamma$ is a finite set called the \emph{output alphabet}, and
\item $f\colon Q\to\Gamma$ is the \emph{output function from final states}.\hfill\IEEEQED
\end{enumerate}
\end{definition}

\vspace*{-1mm}

A finite automaton with outputs computes as follows. 

\begin{definition}\label{Comp_FAO}
Let $M=(Q,\Sigma,\delta,q_0,\Gamma,f)$ be a finite automaton with outputs.
For every $x=x_1x_2\dots x_n\in\Sigma^*$ with each $x_i\in\Sigma$,
the output of $M$ on the input $x$, denoted $M(x)$, is $y\in\Gamma$ for which there exist $q_1,q_2,\dots,q_n\in Q$ such that
\vspace*{-1mm}
\begin{enumerate}
\item $q_{i}=\delta(q_{i-1},x_{i})$ for every $i\in\{1,2,\dots,n\}$, and
\item $y=f(q_n)$.\hfill\IEEEQED
\end{enumerate}
\end{definition}

\vspace*{-1mm}

In Definitions~\ref{def_FAO} and \ref{Comp_FAO},
if we set $\Gamma=\{0,1\}$, the definitions result in those of a normal deterministic finite automaton and its computation,
where $M(x)=1$ means that $M$ accepts $x$ and $M(x)=0$ means that $M$ rejects $x$.

\begin{definition}[Finite-State Strong Predictability]\label{def_fssp}
For any $X\in\XI$, we say that $X$ is
\emph{finite-state strongly predictable} if there exists a finite automaton with outputs
$M=(Q,\X,\delta,q_0,\{0,1,N\},f)$
for which the following two conditions hold:
\vspace*{-1mm}
\begin{enumerate}
\item For every $n\in\N$, if $M(\rest{X}{n})\neq N$ then
  $M(\rest{X}{n})=X(n+1)$.
\item The set $\{n\in\N\mid M(\rest{X}{n})\neq N\}$ is infinite.\hfill\IEEEQED
\end{enumerate}
\end{definition}

Since the computation of every finite automaton can be simulated by some deterministic Turing machine which always halts, the following holds, obviously.

\begin{proposition}\label{fact_fssp_tsp}
For every $X\in\XI$, if $X$ is finite-state strongly predictable then $X$ is total strongly predictable.
\hfill\IEEEQED
\end{proposition}

\begin{theorem}\label{main_fssp}
Let $T$ be a real with $0<T<1$.
For every $X\in\XI$, if $X$ is strictly Chaitin $T$-random, then $X$ is finite-state strongly predictable.
\hfill\IEEEQED
\end{theorem}

\vspace*{-1mm}

In order to prove Theorem~\ref{main_fssp} we need the following theorem. For completeness, we include its proof.

\begin{theorem}[Calude, Hay, and Stephan \cite{CHS11}]\label{bounded-run}
Let $T$ be a real with $0<T<1$.
For every $X\in\XI$, if $X$ is strictly Chaitin $T$-random,
then
there exists $L\ge 2$ such that $X$ does not have a run of $L$ consecutive zeros.
\end{theorem}

\vspace*{-1mm}

\begin{proof}
Based on the optimality of $U$ used in the definition of $H(x)$, it is easy to show that there exists $d\in\N$ such that, for every $x\in\X$ and every $n\in\N$,
\vspace*{-0.5mm}
\begin{equation}\label{HVs01nleHs+Hn+d}
  H(x0^n)\le H(x)+H(n)+d.
\end{equation}
\vspace*{-4.0mm}\\
Since $T>0$, it follows
also from the optimality of $U$ that there exists $c\in\N^+$ such that $H(c)+d\le Tc-1$.
Hence, by \eqref{HVs01nleHs+Hn+d} we see that,
for every $x\in\X$,
\vspace*{-0.5mm}
\begin{equation}\label{Hs0cleHs+Tc-1}
  H(x0^c)\le H(x)+Tc-1.
\end{equation}

\vspace*{-2mm}

Now, suppose that $X\in\XI$ is strictly Chaitin $T$-random.
Then there exists $d_0\in\N$ such that, for every $n\in\N$,
\begin{equation}\label{Hbn-Tnld0}
  \abs{H(\rest{X}{n})-Tn}\le d_0.
\end{equation}
We choose a particular $k_0\in\N^+$ with $k_0>2d_0$.

Assume that $X$ has a run of $ck_0$ consecutive zeros.
Then $\rest{X}{n_0}0^{ck_0}=\rest{X}{n_0+ck_0}$ for some $n_0\in\N$.
It follows from \eqref{Hs0cleHs+Tc-1} that
$H(\rest{X}{n_0+ck_0})-T(n_0+ck_0)+k_0\le H(\rest{X}{n_0})-Tn_0$.
Thus, using \eqref{Hbn-Tnld0} we have $-d_0+k_0\le d_0$, which contradicts the fact that $k_0>2d$.
Hence, $X$ does not have a run of $ck_0$ consecutive zeros, as desired.
\end{proof}

\begin{proof}[Proof of Theorem~\ref{main_fssp}]
Suppose that $X\in\XI$ is strictly Chaitin $T$-random.
Then, by Theorem~\ref{bounded-run},
there exists $d\ge 2$ such that $X$ does not have a run of $d$ consecutive zeros.
For each $n\in\N^+$, let $a(n)$ be the length of the $n$th block of consecutive zeros in $X$ from the left.
Namely, assume that $X$ has the form
$X=1^{b(0)}0^{a(1)}1^{b(1)}0^{a(2)}1^{b(2)}0^{a(3)}1^{b(3)}\dotsm\dotsm$
for some natural number $b(0)$ and some infinite sequence $b(1),b(2),b(3),\dotsc$ of positive integers.
Let
$L=\limsup_{n\to\infty}a(n)$.
Since $1\le a(n)<d$ for all $n\in\N^+$, we have $L\in\N^+$.
Moreover, since $\{a(n)\}$ is a sequence of positive integers, there exists $n_0\in\N^+$ such that
\vspace*{-1.0mm}
\begin{equation}\label{Lp1}
  a(n)\le L
\end{equation}
\vspace*{-4.5mm}\\
for every $n\ge n_0$, and
\vspace*{-2.5mm}
\begin{equation}\label{Lp2}
  a(n)=L
\end{equation}
\vspace*{-3.5mm}\\
for infinitely many $n\ge n_0$.
Let $m$ be the length of the prefix of $X$ which lies immediately to the left of the $n_0$th block of consecutive zeros in $X$.
Namely,
$m=\sum_{k=0}^{n_0-1}b(k)+\sum_{k=1}^{n_0-1}a(k)$.

Now, we define a finite automaton with outputs $M=(Q,\X,\delta,q_0,\{0,1,N\},f)$ as follows:
First, $Q$ is defined as $\{q_0,q_1,\dots,q_{m+L}\}$.
The transition function $\delta$ is then defined by
\vspace*{-3.0mm}
\begin{align*}
\delta(q_{i},0)&=\delta(q_{i},1)=q_{i+1}\quad\text{ if }i=0,\dots,m-1,\\
\delta(q_{i},0)&=q_{i+1}\quad\text{ if }i=m,\dots,m+L-1,\\
\delta(q_{i},1)&=q_{m}\quad\text{ if }i=m,\dots,m+L,
\end{align*}
\vspace*{-4.0mm}\\
where $\delta(q_{m+L},0)$ is arbitrary.
Finally, the output function $f\colon Q\to\{0,1,N\}$ is defined by
$f(q)=1$ if $q=q_{m+L}$ and $N$ otherwise.

\vspace*{-2mm}

Then, it is easy to see that, for every $x\in\X$,
\vspace*{-1mm}
\begin{enumerate}
\item $M(x)=1$ if and only if there exists $y\in\X$ such that $\abs{y}\ge m$ and $x=y0^L$, and
\item $M(x)\neq 0$.
\end{enumerate}

Now, for an arbitrary $n\in\N$, assume that $M(\rest{X}{n})\neq N$.
Then, by the condition (ii) above, we have $M(\rest{X}{n})=1$.
Therefore, by the condition (i) above, there exists $y\in\X$ such that $\abs{y}\ge m$ and $\rest{X}{n}=y0^L$.
It follows from \eqref{Lp1} that $X(n+1)=1$ and therefore $M(\rest{X}{n})=X(n+1)$.
Thus the condition (i) of Definition~\ref{def_fssp} holds for $M$ and $X$.
On the other hand, using \eqref{Lp2} and the condition (i) above, it is easy to see that the set $\{n\in\N\mid M(\rest{X}{n})=1\}$ is infinite.
Thus the condition (ii) of Definition~\ref{def_fssp} holds for $M$ and $X$.
Hence, $X$ is finite-state strongly predictable. 
\end{proof}

\begin{theorem}\label{main1}
Let $T$ be a computable real with $0<T<1$. Then $Z(T)$ is finite-state strongly predictable.
\end{theorem}

\vspace*{-1mm}

\begin{proof}
The result follows immediately from (i) of Theorem~\ref{ZVTwCTr-strictTcb} and Theorem~\ref{main_fssp}.
\end{proof}

In the case where $T$ is a computable real with $0<T<1$, $Z(T)$ is not computable despite Theorem~\ref{main1}.  
This is because, in such a case, $Z(T)$ is weakly Chaitin $T$-random by (i) of Theorem~\ref{ZVTwCTr-strictTcb},
and therefore $Z(T)$ cannot be computable.

It is worthwhile to investigate the behavior of $Z(T)$ in the case where $T$ is not computable but $0<T<1$.
On the one hand, note that $Z(T)$ is of class $C^{\infty}$ as a function of $T\in(0,1)$ \cite{T02} and $\frac{d}{dT}Z(T)>0$ for every $T\in(0,1)$.
On the other hand, recall that a real is weakly Chaitin random almost everywhere.
Thus, by Theorem~\ref{wCrnonsp}, we have $\mathcal{L}(S)=1$,
where $S$ is the set of all $T\in(0,1)$ such that $T$ is not computable and $Z(T)$ is not strongly predictable, and $\mathcal{L}$ is Lebesgue measure on $\R$.

\vspace*{-1.0mm}

\section*{Acknowledgments}
The author is grateful to
Professor
Cristian S. Calude
for his
encouragement.
This work was supported by JSPS KAKENHI Grant Number 23340020.

\vspace*{-2mm}

\end{document}